\documentclass[superscriptaddress,aps,pra,twocolumn,showpacs,nofootinbib,longbibliography]{revtex4-1}
\usepackage{etex}
\usepackage{amsmath,amssymb,amsthm}
\usepackage[colorlinks=true,citecolor=blue,urlcolor=blue]{hyperref}
\usepackage[pdftex]{graphicx}
\usepackage{times,txfonts}
\usepackage{braket}
\usepackage{color}
\usepackage{natbib}
\usepackage{amsmath,blkarray}
\usepackage{mathtools}
\usepackage{latexsym}
\usepackage{tabularx, booktabs}
\usepackage{graphics,epstopdf}
\usepackage{graphicx}
\usepackage{float}
\usepackage{graphicx}
\usepackage{amsfonts}
\usepackage{subcaption}

\newcommand{\be}{\begin{equation}}
\newcommand{\ee}{\end{equation}}
\newcommand{\ba}{\begin{eqnarray}}
\newcommand{\ea}{\end{eqnarray}}

\newcommand{\tr}{\operatorname{Tr}}

\newtheorem{observation}{Observation}

\newtheorem{thm}{Theorem}

\begin{document}

\title{Detecting Einstein-Podolsky-Rosen steering through entanglement detection}

\author{Debarshi Das}
\email{dasdebarshi90@gmail.com}
\affiliation{Centre for Astroparticle Physics and Space Science (CAPSS),
Bose Institute, Block EN, Sector V, Salt Lake, Kolkata 700 091, India}

\author{Souradeep Sasmal}
\email{souradeep@mail.jcbose.ac.in}
\affiliation{Centre for Astroparticle Physics and Space Science (CAPSS),
Bose Institute, Block EN, Sector V, Salt Lake, Kolkata 700 091, India}

\author{Sovik Roy}
\email{sovik1891@gmail.com}
\affiliation{Department of Mathematics, EM 4/1, Techno India Salt Lake (TISL), Sector V, Kolkata 700 091, India}

\begin{abstract}
Quantum inseparabilities can be classified into three inequivalent forms: entanglement, Einstein-Podolsky-Rosen (EPR) steering, and Bell's nonlocality. Bell-nonlocal states form a strict subset of EPR steerable states which also form a strict subset of entangled states. Recently, EPR steerable states are shown to be fundamental resources for one-sided device-independent quantum information processing tasks and, hence, identification of EPR steerable states becomes important from foundational as well as informational theoretic perspectives. In the present study we propose a new criteria to detect whether a given two-qubit state is EPR steerable. From an arbitrary given two-qubit state, another two-qubit state is constructed in such a way that the given state is EPR steerable if the new constructed state is entangled. Hence, EPR steerability of an arbitrary two-qubit state can be detected by detecting entanglement of the newly constructed state. Apart from providing a distinctive way to detect EPR steering without using any steering inequality, the novel finding in the present study paves a new direction to avoid locality loophole in EPR steering tests and to reduce the ``complexity cost" present in experimentally detecting EPR steering.  We also generalise our criteria to detect EPR steering of higher dimensional quantum states. Finally, we illustrate our result by using our proposed technique to detect EPR steerability of various families of mixed states.
\end{abstract}

\maketitle
\section{Introduction}

Einstein-Podolsky-Rosen (EPR) steering is defined as the apparent ability to affect a spatially separated  quantum state, which was the central problem in the EPR argument \cite{epr} to demonstrate the incompleteness of quantum mechanics. In particular, EPR argument considers an entangled state shared between two spatially separated parties and it implies the possibility to produce different set of states at one party's end by performing local quantum measurements of any two non-commuting observables on another spatially separated party's end. This ``Spooky action at a distance" motivated Schr\"{o}dinger to conceive the celebrated concept of `EPR steering'  \cite{scro}. However, the research field of quantum steering did not progress much until 2007, when Wiseman, Jones, and Doherty (WJD) introduced the concept of EPR steering in the form of a task \cite{steer,steer2}. The task of quantum steering is that a referee has to determine (using the measurement outcomes communicated classically from the two parties to the referee) whether two spatially separated parties share entanglement, when one of the two parties is untrusted. WJD  introduced the notion of EPR steering as the inability to construct a local hidden variable-local hidden state (LHV-LHS) model to explain the joint probabilities of measurement outcomes. Note that in EPR steering scenario the no-signalling condition (the probability of obtaining one party's outcome does not depend on spatially separated other party's measurement setting) is always satisfied by the spatially separated two parties.

It is well-known that EPR steering is an intermediate form of quantum inseparabilities in between entanglement \cite{ent} and Bell nonlocality \cite{Bell,chsh,bell2}. Quantum states that demonstrate Bell-nonlocality form a strict subset of quantum states demonstrating EPR steering which also form a strict subset of entangled states \cite{inequi}. One important point to be stressed here is that EPR steering is inherently asymmetric with respect to the observers  unlike quantum nonlocality and entanglement \cite{one_way_steer}. In this case, the outcome statistics of one subsystem (which is being `steered') is produced due to quantum measurements on a quantum state. However, there is no such constraint for the other subsystem. In fact it is shown that there exist entangled states which are one-way steerable, i.e., demonstrate steerability from one observer to the other spatially separated observer, but not vice-versa \cite{one_way_steer,one2,one3}. Apart from having foundational significance, EPR steering has a vast range of information theoretic application in one-sided device-independent scenario where the party, which is being steered, has trust on his/her quantum device but the other party's device is untrusted. These applications range from one-sided device-independent quantum key distribution \cite{app1}, advantages in sub-channel discrimination \cite{app2}, secure quantum teleportation \cite{app3,app8}, quantum communication \cite{app8},  detecting bound entanglement \cite{app10}, one-sided device-independent randomness generation \cite{app9}, one-sided device-independent self-testing of pure maximally as well as non-maximally entangled states \cite{app14}.

Against this above backdrop, from fundamental viewpoint as well as from information theoretic perspective it is important to detect EPR steerable states. A number of criteria to detect EPR steering have been proposed till date \cite{in1,in4,in5,in6,in10,in8,nin18,nin10,in13,in17,nin17}. In the present study we provide a completely new criteria to detect EPR steering of an arbitrary two-qubit state. Given an arbitrary two-qubit state, a new two-qubit state is constructed in such a way that EPR-steering of the given state is detected if the new constructed state is entangled. Hence, following Peres-Horodecki criteria \cite{ph1,ph2} we can state that the given state is EPR-steerable if the partial transpose of the new constructed state has at least one negative eigenvalue.  We present a brief idea on how our result can be implemented in experiment. A possible extension of our result to the case of higher dimensional quantum systems has also been demonstrated.

The present study indicates a connection between EPR steering and entanglement. The novelty of the result obtained in the present study is that it presents a method to detect EPR steering without using any steering inequality. From experimental point of view, the present study enables one to indirectly test EPR steering of an arbitrary two-qubit state through entanglement witness \cite{ew1,ew2,ew3,ew4,ew5,ew6} of the new constructed state. Hence, our proposed theorem enables to reduce the ``complexity cost" \cite{cc} (it quantifies how complex an experiment is in order to determine entanglement, EPR steering, Bell nonlocality) in experimentally determining EPR steering as it has been shown that the ``complexity cost" for the least complex demonstration of entanglement is less than the ``complexity cost" for the least complex demonstration of EPR steering \cite{cc}. In particular, new nonlocality tests were constructed to demonstrate the complexity costs of entanglement and EPR steering. These inequalities are the simplest possible witnesses for the above two types of quantum inseparabilities. On the other hand, Bell-CHSH (Bell-Clauser-Horne-Shimony-Holt) inequality is the simplest possible witness for Bell nonlocality.  The above complexity costs have also been demonstrated experimentally by showing the violations of these new inequalities and Bell-CHSH inequality using photonic singlet states \cite{cc}. 

The present study may be helpful to avoid the locality loophole present in EPR steering test. Because the degree of correlation required for entanglement testing is smaller than that for violation of a steering inequality, it should be correspondingly easier to demonstrate entanglement without making the fair-sampling assumption \cite{sbl}. Our proposed procedure to test EPR steering through entanglement detection makes experimental demonstration of EPR steering easier since demonstrating EPR steering is strictly harder than demonstrating entanglement as mentioned in Ref. \cite{sbl}. One important point to be stressed here is that Bell nonlocality can be indirectly detected by detecting EPR steering \cite{pati}. The present study, therefore, completes demonstrating the connections between three inequivalent forms of quantum inseparabilities.

We organize this paper in the following way. We briefly discuss the concept of EPR steering and entanglement in Section \ref{sec2}. In Section \ref{sec3}, we present the main result of this paper on detecting EPR steering of an arbitrary two-qubit state indirectly through entanglement detection. Whether this result can be extended to higher dimensional system is also discussed in Section \ref{sec3}. We illustrate our result by detecting EPR steerability  of  various classes of two-qubit mixed states  and qubit-qutrit mixed states using our proposed technique in Section \ref{sec4}. Finally, in Section \ref{sec5} we summarize the results obtained and present the concluding remarks.

\section{Preliminaries}\label{sec2}
 Suppose $A \in \mathbb{F}_{\alpha} $ and $B \in \mathbb{F}_{\beta}$ are the possible choices of measurements for two spatially separated observers, say Alice and Bob, with outcomes $a \in \mathbb{G}_{a}$ and $b \in \mathbb{G}_{b}$, respectively. Let the state $\rho_{AB}$ is shared between Alice and Bob. After Alice performs arbitrary measurement $A$ with measurement operators $M^A_{a}$ ($M^A_a \geq 0$ $\forall \ A, a$ and $\sum_a M^A_{a} = \mathbb{I}$ $\forall \ A$) corresponding to the outcome $a$, Bob's (unnormalized) conditional state becomes
\begin{equation}
\sigma^A_a = \tr_A[(M^A_{a} \otimes \mathbb{I}) \ \rho_{AB}],
\end{equation}
where $\mathbb{I}$ is the $2 \times 2$ identity matrix. On this unnormalized conditional state Bob performs measurement $B$ with measurement operators $M^B_{b}$ corresponding to the outcome $b$ ($M^B_b \geq 0$ $\forall \ \ B, b$ and $\sum_b M^B_{b} = \mathbb{I}$ $\forall \ B$) to produce the joint probability distribution $P(a, b|A, B, \rho_{AB}) = \tr[M^B_{b} \sigma^A_a]$, where $P(a, b|A, B, \rho_{AB})$ denotes the joint probability of obtaining the outcomes $a$ and $b$, when measurements $A$ and $B$ are performed by Alice and Bob locally on state $\rho_{AB}$, respectively.

The bipartite state $\rho_{AB}$ of the system is steerable from Alice to Bob if and only if it is not the case that for all  $A \in \mathbb{F}_{\alpha} $, $B \in \mathbb{F}_{\beta}$, $a \in \mathbb{G}_{a}$, $b \in \mathbb{G}_{b}$, the joint probability distribution can be written in the form
\begin{equation}
P(a, b|A, B, \rho_{AB}) = \sum_{\lambda} P(\lambda ) \ P(a|A,\lambda) \ P_Q(b|B,\rho_{\lambda}),
\label{lhv-lhs}
\end{equation}
where $P(\lambda)$ is the probability distribution over the local hidden variables (LHV) $\lambda$ with $\sum_{\lambda} P(\lambda ) = 1$. $P(a|A,\lambda)$ denotes an arbitrary probability distribution and $P_Q(b|B,\rho_{\lambda})(=\tr [\rho_{\lambda} M^B_b])$ denotes the quantum probability of outcome $b$ given measurement $B$ on the local hidden state (LHS) $\rho_{\lambda}$; $M^B_b$ being the measurement operator of the observable $B$ associated with outcome $b$. In other words, the bipartite state $\rho_{AB}$ is steerable from Alice to Bob if and only if it does not have LHV-LHS model description (\ref{lhv-lhs}) for arbitrary measurements performed by Alice and Bob.

A bipartite state $\rho_{AB}$ is called separable if and only if the state can be written in the following form
\begin{equation}
\rho_{AB} = \sum_{\lambda} P(\lambda ) \ \rho^A_{\lambda} \otimes \rho^B_{\lambda},
\end{equation}
where $\sum_{\lambda} P(\lambda ) = 1$. A bipartite state, which is not separable, is called entangled. Alternatively, the bipartite state $\rho_{AB}$ of the system is entangled if and only if it is not the case that for all  $A \in \mathbb{F}_{\alpha} $, $B \in \mathbb{F}_{\beta}$, $a \in \mathbb{G}_{a}$, $b \in \mathbb{G}_{b}$, the joint probability distribution can be written in the form
\begin{equation}
P(a, b|A, B, \rho_{AB}) = \sum_{\lambda} P(\lambda ) \ P_Q(a|A,\rho^A_{\lambda}) \ P_Q(b|B,\rho^B_{\lambda}),
\label{lhs-lhs}
\end{equation}
where $\sum_{\lambda} P(\lambda ) = 1$. $P_Q(a|A,\rho^A_{\lambda})(=$ tr$[\rho^A_{\lambda} M^A_a])$ denotes the quantum probability of outcome $a$ given measurement $A$ on the local hidden state $\rho^A_{\lambda}$; $M^A_a$ being the measurement operator of the observable $A$ associated with outcome $a$. $P_Q(b|B,\rho^B_{\lambda})$ is similarly defined. In other words, the bipartite state $\rho_{AB}$ is entangled if and only if it does not have LHS-LHS model description (\ref{lhs-lhs}) for arbitrary measurements performed by Alice and Bob.

\section{Detecting EPR steering through Entanglement detection}\label{sec3}

The main result of this paper is stated in the following theorem.

\begin{thm}
For any two-qubit state $\rho_{AB}$ shared between Alice and Bob, define two new states $\tau^1_{AB}$ and $\tau^2_{AB}$ given by,
\begin{equation}
\tau^1_{AB} = \mu_1 \ \rho_{AB} + (1-\mu_1) \ \tilde{\rho^1}_{AB},
\end{equation}
and 
\begin{equation}
\tau^2_{AB} = \mu_2 \ \rho_{AB} + (1-\mu_2) \ \tilde{\rho^2}_{AB},
\end{equation}
where $\tilde{\rho^1}_{AB} = \rho_A \otimes \frac{\mathbb{I}}{2}$ with $\rho_A = \tr_{B}[ \rho_{AB} ] = \tr_{B}[ \tau^1_{AB} ]$ being the reduced state at Alice's side; $\tilde{\rho^2}_{AB} = \frac{\mathbb{I}}{2} \otimes \rho_B$ with $\rho_B = \tr_{A}[ \rho_{AB} ] = \tr_{A}[ \tau^2_{AB} ]$ being the reduced state at Bob's side; $\mu_1 \in [0, \frac{1}{\sqrt{3}}]$; $\mu_2 \in [0, \frac{1}{\sqrt{3}}]$. If $\tau^1_{AB}$ is entangled, then $\rho_{AB}$ is EPR steerable from Bob to Alice. On the other hand, if $\tau^2_{AB}$ is entangled, then $\rho_{AB}$ is EPR steerable from Alice to Bob.
\label{th1}
\end{thm}

\begin{proof}
At first we shall prove that if $\tau^1_{AB}$ is entangled, then $\rho_{AB}$ is EPR steerable from Bob to Alice. We shall prove this by proving its converse negative proposition: if $\rho_{AB}$ is not EPR steerable from Bob to Alice, then $\tau^1_{AB}$ is separable.

Let us calculate $P(a, b|A, B, \tau^1_{AB}) = \tr[M^B_{b} \varsigma^A_a]$, which is the joint probability of obtaining the outcomes $a$ and $b$, when arbitrary measurements $A$ and $B$ are performed by Alice and Bob locally on state $\tau^1_{AB}$, respectively, where $ \varsigma^A_a$ is the (unnormalized) conditional state on Bob's side when Alice performs measurement $A$ with measurement operators $M^A_{a}$ corresponding to the outcome $a$.
\begin{align}
\varsigma^A_a &= \tr_A\Big[\big(M^A_{a} \otimes \mathbb{I}\big) \tau^1_{AB}\Big] \nonumber \\
& =  \tr_A\Big[\big(M^A_{a} \otimes \mathbb{I}\big) \big( \mu_1 \rho_{AB} + (1-\mu_1) \tilde{\rho^1}_{AB}\big)\Big] \nonumber \\
& = \mu_1  \tr_A\Big[\big(M^A_{a} \otimes \mathbb{I}\big) \rho_{AB}\Big] + (1-\mu_1) \tr_A\Big[\big(M^A_{a} \otimes \mathbb{I}\big) \tilde{\rho^1}_{AB}\Big] \nonumber \\
& = \mu_1  \tr_A\Big[\big(M^A_{a} \otimes \mathbb{I}\big) \rho_{AB}\Big] + (1-\mu_1) P(a|A, \rho_{AB}) \frac{\mathbb{I}}{2},
\label{eedd}
\end{align}
where $P(a|A, \rho_{AB})$ denotes the marginal probability of Alice to obtain the outcome $a$ contingent upon performing measurement $A$ on the state $\rho_{AB}$.

Let us assume that
\begin{align}
\varsigma^A_a &= \begin{pmatrix}
n_1 && n_2\\
n_3 && n_4\\
\end{pmatrix} \nonumber \\
&= \frac{n_1 + n_4}{2} \mathbb{I} + \text{Re}[n_2] \sigma_x - \text{Im}[n_2] \sigma_y + \frac{n_1 - n_4}{2} \sigma_z,
\label{cons}
\end{align}
where $n_1$, $n_4$ are real and $n_2$ = $\bar{n}_3$ with $\bar{n}_3$ being the complex conjugate of $n_3$ (since $\varsigma^A_a$ is a Hermitian matrix). $\sigma_x$, $\sigma_y$, $\sigma_z$ are Pauli matrices. $\text{Re}[n_2]$ and $\text{Im}[n_2]$ denote the real and imaginary part of $n_2$, respectively.

Now let us evaluate the matrix elements of $\varsigma^A_a$ using Eq. (\ref{eedd}).
\begin{align}
n_1 &= \tr \Big[\pi^{\hat{z}}_{+} \varsigma^A_a\Big] \nonumber \\
& = \mu_1 P(a, +| A, \hat{z}, \rho_{AB}) + (1-\mu_1) P(a|A, \rho_{AB}) \frac{1}{2},
\label{n_1}
\end{align}
where $\pi^{\hat{z}}_{+}$ is the projector onto the eigenstate of $\sigma_z$ corresponding to the eigenvalue $+1$ and it is given by,
\begin{equation}
\pi^{\hat{z}}_{+} = \begin{pmatrix}
1 && 0\\
0 && 0\\
\end{pmatrix}.
\end{equation}
$P(a, +| A, \hat{z}, \rho_{AB})$ is the joint probability of obtaining the outcomes $a$ and $+1$, when measurement $A$ and projective measurement corresponding to the operator $\sigma_z$ are performed by Alice and Bob locally on state $\rho_{AB}$, respectively.

\begin{align}
n_4 &= \tr \Big[\pi^{\hat{z}}_{-} \varsigma^A_a\Big] \nonumber \\
& = \mu_1 P(a, -| A, \hat{z}, \rho_{AB}) + (1-\mu_1) P(a|A, \rho_{AB}) \frac{1}{2},
\label{n_4}
\end{align}
where $\pi^{\hat{z}}_{-}$ is the projector onto the eigenstate of $\sigma_z$ corresponding to the eigenvalue $-1$ and it is given by,
\begin{equation}
\pi^{\hat{z}}_{-} = \begin{pmatrix}
0 && 0\\
0 && 1\\
\end{pmatrix}.
\end{equation}
$P(a, -| A, \hat{z}, \rho_{AB})$ is the joint probability of obtaining the outcomes $a$ and $-1$, when measurement $A$ and projective measurement corresponding to the operator $\sigma_z$ are performed by Alice and Bob locally on state $\rho_{AB}$, respectively.
Hence, from Eqs.(\ref{n_1}) and (\ref{n_4}) we get
\begin{equation}
n_1 + n_4 = P(a|A, \rho_{AB})
\label{plus}
\end{equation}
and 
\begin{equation}
n_1 - n_4 = 2 \mu_1 P(a, +| A, \hat{z}, \rho_{AB}) - \mu_1 P(a|A, \rho_{AB}).
\label{minus}
\end{equation}

\begin{align}
\text{Re}[n_2] &= \tr \Big[\pi^{\hat{x}}_{+} \varsigma^A_a\Big] - \frac{1}{2} P(a|A, \rho_{AB}) \nonumber \\
& = \mu_1 P(a, +| A, \hat{x}, \rho_{AB}) -\frac{\mu_1}{2} P(a|A, \rho_{AB}),
\label{Rn_2}
\end{align}
where $\pi^{\hat{x}}_{+}$ is the projector onto the eigenstate of $\sigma_x$ corresponding to the eigenvalue $+1$ and it is given by,
\begin{equation}
\pi^{\hat{x}}_{+} = \begin{pmatrix}
\dfrac{1}{2} && \dfrac{1}{2}\\[10pt]
\dfrac{1}{2} && \dfrac{1}{2}\\
\end{pmatrix}.
\end{equation}
$P(a, +| A, \hat{x}, \rho_{AB})$ is the joint probability of obtaining the outcomes $a$ and $+1$, when measurement $A$ and projective measurement corresponding to the operator $\sigma_x$ are performed by Alice and Bob locally on state $\rho_{AB}$, respectively.

\begin{align}
\text{Im}[n_2] &= - \tr \Big[\pi^{\hat{y}}_{+} \varsigma^A_a\Big] + \frac{1}{2} P(a|A, \rho_{AB}) \nonumber \\
& = - \mu_1 P(a, +| A, \hat{y}, \rho_{AB}) + \frac{\mu_1}{2} P(a|A, \rho_{AB}),
\label{In_2}
\end{align}
where $\pi^{\hat{y}}_{+}$ is the projector onto the eigenstate of $\sigma_y$ corresponding to the eigenvalue $+1$ and it is given by
\begin{equation}
\pi^{\hat{y}}_{+} = \begin{pmatrix}
\dfrac{1}{2} && \dfrac{-i}{2}\\[10pt]
\dfrac{i}{2} && \dfrac{1}{2}\\
\end{pmatrix}.
\end{equation}
$P(a, +| A, \hat{y}, \rho_{AB})$ is the joint probability of obtaining the outcomes $a$ and $+1$, when measurement $A$ and projective measurement corresponding to the operator $\sigma_y$ are performed by Alice and Bob locally on state $\rho_{AB}$, respectively.

Combining Eqs. (\ref{cons}), (\ref{plus}), (\ref{minus}), (\ref{Rn_2}), (\ref{In_2}) we obtain
\begin{align}
\varsigma^A_a &= P(a|A, \rho_{AB}) \frac{\mathbb{I}}{2} \nonumber \\
&+ \Big(\mu_1 P(a, +| A, \hat{x}, \rho_{AB}) -\frac{\mu_1}{2} P(a|A, \rho_{AB}) \Big) \sigma_x \nonumber \\
&+ \Big(\mu_1 P(a, +| A, \hat{y}, \rho_{AB}) - \frac{\mu_1}{2} P(a|A, \rho_{AB}) \Big) \sigma_y \nonumber \\
&+ \Big(\mu_1 P(a, +| A, \hat{z}, \rho_{AB}) - \frac{\mu_1}{2} P(a|A, \rho_{AB}) \Big) \sigma_z,
\label{cons2}
\end{align}
Therefore, from Eq.(\ref{cons2}) we get
\begin{align}
P(a, b|A, B, \tau^1_{AB}) &= \tr[M^B_{b} \varsigma^A_a] \nonumber \\
&= \tr\Bigg[ M^B_{b} \Big\{P(a|A, \rho_{AB}) \frac{\mathbb{I}}{2} \nonumber \\
&+ \mu_1 P(a, +| A, \hat{x}, \rho_{AB}) \sigma_x -\frac{\mu_1}{2} P(a|A, \rho_{AB}) \sigma_x \nonumber \\
&+ \mu_1 P(a, +| A, \hat{y}, \rho_{AB}) \sigma_y - \frac{\mu_1}{2} P(a|A, \rho_{AB}) \sigma_y \nonumber \\
&+ \mu_1 P(a, +| A, \hat{z}, \rho_{AB}) \sigma_z - \frac{\mu_1}{2} P(a|A, \rho_{AB}) \sigma_z \Big\} \Bigg].
\label{jtpr}
\end{align}
Now, if $\rho_{AB}$ is not steerable from Bob to Alice, then for all  $A \in \mathbb{F}_{\alpha} $, $B \in \mathbb{F}_{\beta}$, $a \in \mathbb{G}_{a}$, $b \in \mathbb{G}_{b}$, the joint probability distribution can be written in the form
\begin{equation}
P(a, b|A, B, \rho_{AB}) = \sum_{\lambda} P(\lambda )P_Q(a|A,\rho^A_{\lambda})P(b|B,\lambda).
\label{lhvlhsf}
\end{equation}
We can thus write down the following,
\begin{equation}
P(a, +|A, \hat{x}, \rho_{AB}) = \sum_{\lambda} P(\lambda )P_Q(a|A,\rho^A_{\lambda})P(+|\hat{x},\lambda),
\label{eq1}
\end{equation}
\begin{equation}
P(a, +|A, \hat{y}, \rho_{AB}) = \sum_{\lambda} P(\lambda )P_Q(a|A,\rho^A_{\lambda})P(+|\hat{y},\lambda),
\label{eq2}
\end{equation}
\begin{equation}
P(a, +|A, \hat{z}, \rho_{AB}) = \sum_{\lambda} P(\lambda )P_Q(a|A,\rho^A_{\lambda})P(+|\hat{z},\lambda),
\label{eq3}
\end{equation}
and 
\begin{equation}
P(a|A, \rho_{AB}) = \sum_{\lambda} P(\lambda )P_Q(a|A,\rho^A_{\lambda}).
\label{eq4}
\end{equation}
Using the above Eqs. (\ref{jtpr}), (\ref{eq1}), (\ref{eq2}), (\ref{eq3}) and (\ref{eq4}) we get
\begin{align}
P(a, b|A, B, \tau^1_{AB}) &= \tr\Bigg[ M^B_{b} \Bigg\{ \Big(\sum_{\lambda} P(\lambda )P_Q(a|A,\rho^A_{\lambda})\Big) \frac{\mathbb{I}}{2} \nonumber \\
&+ \mu_1 \Big(\sum_{\lambda} P(\lambda )P_Q(a|A,\rho^A_{\lambda})P(+|\hat{x},\lambda) \Big) \sigma_x  \nonumber \\
&- \frac{\mu_1}{2} \Big(\sum_{\lambda} P(\lambda )P_Q(a|A,\rho^A_{\lambda})\Big) \sigma_x \nonumber \\
&+ \mu_1 \Big(\sum_{\lambda} P(\lambda )P_Q(a|A,\rho^A_{\lambda})P(+|\hat{y},\lambda) \Big) \sigma_y 
\nonumber \\
&- \frac{\mu_1}{2} \Big(\sum_{\lambda} P(\lambda )P_Q(a|A,\rho^A_{\lambda})\Big) \sigma_y \nonumber \\
&+ \mu_1 \Big(\sum_{\lambda} P(\lambda )P_Q(a|A,\rho^A_{\lambda})P(+|\hat{z},\lambda) \Big) \sigma_z 
\nonumber \\
&- \frac{\mu_1}{2} \Big(\sum_{\lambda} P(\lambda )P_Q(a|A,\rho^A_{\lambda})\Big) \sigma_z \Bigg\} \Bigg].
\label{jtpr2}
\end{align}
Now let us choose 
\begin{equation}
\rho^B_{\lambda} = \dfrac{\mathbb{I} + \vec{\sigma} \cdot \vec{r}_{\lambda}}{2},
\end{equation}
where $\vec{\sigma}$ = $(\sigma_x , \sigma_y, \sigma_z)$ is a vector composed of Pauli matrices and 
\begin{equation}
\vec{r}_{\lambda} = \mu_1 \Big( 2 P(+|\hat{x}, \lambda) - 1, 2 P(+|\hat{y}, \lambda) - 1, 2 P(+|\hat{z}, \lambda) - 1 \Big).
\end{equation}
To be a valid probability distribution, $P(+|\hat{x}, \lambda)$, $P(+|\hat{y}, \lambda)$, $P(+|\hat{z}, \lambda)$ $\in$ $[0, 1]$. It can be easily checked that $|\vec{r}_{\lambda}| \leq 1$ implies $\mu_1 \in \Bigg[0, \dfrac{1}{\sqrt{3}}\Bigg]$. Therefore, $\rho^B_{\lambda}$ is a valid quantum state (qubit) for $\mu_1 \in \Bigg[0, \dfrac{1}{\sqrt{3}}\Bigg]$.

From the above construction one can write down the following for all  $A \in \mathbb{F}_{\alpha} $, $B \in \mathbb{F}_{\beta}$, $a \in \mathbb{G}_{a}$, $b \in \mathbb{G}_{b}$,
\begin{align}
&\sum_{\lambda} P(\lambda )P_Q(a|A,\rho^A_{\lambda})P_Q(b|B,\rho^B_{\lambda}) \nonumber\\
&= \sum_{\lambda} P(\lambda )P_Q(a|A,\rho^A_{\lambda})\tr\Bigg[ M^B_{b} \dfrac{\mathbb{I} + \vec{\sigma} \cdot \vec{r}_{\lambda}}{2} \Bigg] \nonumber\\
&= \sum_{\lambda} \Bigg( P(\lambda ) P_Q(a|A,\rho^A_{\lambda}) \Bigg) \tr \Bigg[M^B_{b} \Big( \frac{\mathbb{I}}{2} + \mu_1 P(+|\hat{x}, \lambda) \sigma_x \nonumber \\
&- \frac{\mu_1}{2} \sigma_x + \mu_1 P(+|\hat{y}, \lambda) \sigma_y - \frac{\mu_1}{2} \sigma_y + \mu_1 P(+|\hat{z}, \lambda) \sigma_z - \frac{\mu_1}{2} \sigma_z \Big) \Bigg].
\label{rhs}
\end{align}
Comparing Eqs. (\ref{jtpr2}) and (\ref{rhs}), we can write
\begin{equation}
P(a, b|A, B, \tau^1_{AB}) = \sum_{\lambda} P(\lambda )P_Q(a|A,\rho^A_{\lambda})P_Q(b|B,\rho^B_{\lambda}).
\label{lhslhsf}
\end{equation}
Hence, we can conclude that if for arbitrary measurement $A \in \mathbb{F}_{\alpha}$ performed by Alice and for arbitrary measurement $B \in \mathbb{F}_{\beta}$ performed by Bob the joint probability distribution obtained from the state $\rho_{AB}$ can be written in the form given by Eq.(\ref{lhvlhsf}),
then the joint probability distribution obtained from the state $\tau^1_{AB}$ can always be written in the form given by Eq.(\ref{lhslhsf}). In other words, if $\rho_{AB}$ is not EPR steerable from Bob to Alice, then $\tau^1_{AB}$ is separable.

In a similar way, it can be shown that if $\rho_{AB}$ is not EPR steerable from Alice to Bob, then $\tau^2_{AB}$ is separable. This completes the proof.
\end{proof}

Now we are going to present a brief outline of the possible experimental implementation of Theorem \ref{th1}. Suppose we want to check experimentally whether an arbitrary given two-qubit state $\rho_{AB}$ (shared between Alice and Bob) demonstrates EPR steering from Bob to Alice using Theorem \ref{th1}. In order to achieve this, at first one of the two qubits (say, Bob's qubit) of the given system with state $\rho_{AB}$ is subjected to a local depolarizing channel. Hence, the state after the channel action is given by,
\begin{align}
\rho^f_{AB} = \sum_{i=0}^{3} (\mathbb{I} \otimes K_i) \  \rho_{AB} \  (\mathbb{I} \otimes K_i^\dagger).
\end{align}
Here $K_i$s are the Kraus operators associated with depolarizing channel with $K_0 = \frac{\sqrt{1+3p}}{2} \mathbb{I}$, $K_1 = \frac{\sqrt{1-p}}{2} \sigma_x$, $K_2 = \frac{\sqrt{1-p}}{2} \sigma_y$, $K_3 = \frac{\sqrt{1-p}}{2} \sigma_z$ and $p$ is the channel strength with $0 \leq p \leq 1$. It can be easily checked that 
\begin{equation}
\rho^f_{AB} = p \ \rho_{AB} + (1-p) \ \tilde{\rho}_{AB},
\label{eqfeb1}
\end{equation}
where $\tilde{\rho}_{AB} = \rho_A \otimes \frac{\mathbb{I}}{2}$ with $\rho_A = \tr_{B}[ \rho_{AB} ]$ being the reduced state of the Alice's qubit. 

According to Theorem \ref{th1}, the given state $\rho_{AB}$ is EPR steerable from Bob to Alice if the state $\rho^f_{AB}$ given by Eq.(\ref{eqfeb1}) (with $p \in [0, \frac{1}{\sqrt{3}}]$)  is entangled. 

One can experimentally construct the state $\rho^f_{AB}$ from the given two-qubit state $\rho_{AB}$ by implementing the depolarizing channel following the proposed techniques (see \cite{dc} and the references therein) and then subjecting Bob's qubit to the local depolarizing channel. Finally, entanglement of the state $\rho^f_{AB}$ can be experimentally tested using entanglement witness (\cite{ew1,ew6} and the references therein). 

Similarly, one can experimentally check EPR steering of the given two-qubit state $\rho_{AB}$ from Alice to Bob by subjecting Alice's qubit to a local depolarizing channel, and then checking entanglement of the constructed state using entanglement witness.

According to the famous Peres-Horodecki criteria \cite{ph1,ph2}, any given two-qubit state is entangled if and only if the partial transpose of the given state has at least one negative eigenvalue. Hence, Theorem \ref{th1} immediately provides the following important observation.
\begin{observation}
If the partial transpose of the state $\tau^1_{AB}$ has at least one negative eigenvalue, then $\rho_{AB}$ is EPR steerable from Bob to Alice. On the other hand, if the partial transpose of the state $\tau^2_{AB}$ has at least one negative eigenvalue, then $\rho_{AB}$ is EPR steerable from Alice to Bob; where $\rho_{AB}$, $\tau^1_{AB}$, $\tau^2_{AB}$ are defined in the statement of Theorem \ref{th1}.
\label{obbb1}
\end{observation} 

However, partial transposition is a non-positive map as transposition is not a completely positive map. 
Hence, partial transposition is a non-physical operation. Therefore, one cannot directly implement Observation \ref{obbb1} in the laboratory. However, Horodecki and Ekert proposed a method called ``structural physical approximation" by which non-physical operations such as partial transposition can be systematically approximated by physical operations \cite{spa1}. Interestingly, structural physical approximations to the non-physical operations (including partial transposition) can be factorized into local operations and classical communication \cite{spa2}. Following the structural physical approximation to partial transposition \cite{spa1}, a two-qubit state $\tau_{AB}$ is entangled if and only if the smallest eigenvalue of the state $\tilde{\tau}_{AB}$ is less than $\frac{2}{9}$, where
\begin{equation}
\tilde{\tau}_{AB} = \Lambda(\tau_{AB}) = \frac{2}{9} \ \mathbb{I} \otimes \mathbb{I} + \frac{1}{9} \ [\mathbb{I} \otimes T] (\tau_{AB}).
\end{equation}
Here $T$ is the transposition operation and partial transpose of $\tau_{AB}$ is denoted by $[\mathbb{I} \otimes T] (\tau_{AB})$. Note that the above map $\Lambda(\tau_{AB}) = \tilde{\tau}_{AB}$ is a completely positive map and therefore physically implementable. 
This map can be implemented by applying selected products of unitary (Pauli) transformations with certain probabilities \cite{spa1}. The experimental demonstration of the structural physical approximation to the partial transpose of two-qubit system has been performed in photonic systems using linear optical devices \cite{spa3}. Hence, one can modify observation \ref{obbb1} for experimental implication in the following way: 
\begin{observation}
If the smallest eigenvalue of the state $\tilde{\tau}^1_{AB}$ = $\frac{2}{9} \ \mathbb{I} \otimes \mathbb{I}$ $+$ $\frac{1}{9} \ [\mathbb{I} \otimes T] (\tau^1_{AB})$ is less than $\frac{2}{9}$, then $\rho_{AB}$ is EPR steerable from Bob to Alice. On the other hand, if the smallest eigenvalue of the state $\tilde{\tau}^2_{AB}$ = $\frac{2}{9} \ \mathbb{I} \otimes \mathbb{I}$ $+$ $\frac{1}{9} \ [\mathbb{I} \otimes T] (\tau^2_{AB})$ is less than $\frac{2}{9}$, then $\rho_{AB}$ is EPR steerable from Alice to Bob; where $\rho_{AB}$, $\tau^1_{AB}$, $\tau^2_{AB}$ are defined in the statement of Theorem \ref{th1}; $[\mathbb{I} \otimes T] (\tau^1_{AB})$ and $[\mathbb{I} \otimes T] (\tau^2_{AB})$ denote the partial transpose of $\tau^1_{AB}$ and $\tau^2_{AB}$  respectively. 
\label{obbb2}
\end{observation} 
Hence, in order to experimentally check EPR steering of a given two-qubit state $\rho_{AB}$, one has to experimentally prepare the states $\tau^1_{AB}$ and $\tau^2_{AB}$ from the given state $\rho_{AB}$ by applying local depolarizing channel as described earlier. Then the states $\tilde{\tau}^1_{AB}$ and $\tilde{\tau}^2_{AB}$ described above can be experimentally prepared from $\tau^1_{AB}$ and $\tau^2_{AB}$, respectively, using local measurements and classical communication following the technique adopted in \cite{spa1,spa3}. Finally, the minimum eigenvalue of the states $\tilde{\tau}^1_{AB}$ and $\tilde{\tau}^2_{AB}$ can be estimated by defining a classical optimization problem over the measurement outcomes \cite{spa3} to determine whether the states $\tilde{\tau}^1_{AB}$ and $\tilde{\tau}^2_{AB}$ are entangled.

It is to be noted that the criteria stated in Theorem \ref{th1} to detect EPR steering through entanglement detection is applicable when the shared state belongs to the Hilbert space $\mathcal{H}^2 \otimes \mathcal{H}^2$. However, Theorem \ref{th1} can be generalized to the case when the shared state belongs to the Hilbert space $\mathcal{H}^d \otimes \mathcal{H}^2$, where $d$ is arbitrary. We have proved in Theorem \ref{th1} that if $\rho_{AB}$ is not EPR steerable from Bob to Alice, then $\tau^1_{AB}$ is separable. However, the proof of Theorem \ref{th1} remains unchanged if the dimension of the system at Alice's end is $d$. Note that $M_a^{A}$ stated in the proof of Theorem \ref{th1} is the measurement operator corresponding to the outcome $a$ when Alice performs measurement $A$ on her ``qubit". The whole mathematical proof of Theorem \ref{th1} remains unchanged if the above measurement operator is assumed to be acted on Alice's ``qudit". Thus the proof is independent of the dimension of the system at Alice's end. Hence, we can state the following observation.
\begin{observation}
For any qudit-qubit state $\rho_{AB}$ shared between Alice and Bob, define another new qudit-qubit state $\tau^1_{AB}$ given by
\begin{equation}
\tau^1_{AB} = \mu_1 \ \rho_{AB} + (1-\mu_1) \ \tilde{\rho^1}_{AB},
\end{equation}
where $\tilde{\rho^1}_{AB} = \rho_A \otimes \frac{\mathbb{I}}{2}$ with $\rho_A = \tr_{B}[ \rho_{AB} ] = \tr_{B}[ \tau^1_{AB} ]$ being the reduced state (qudit) at Alice's side; $\mu_1 \in [0, \frac{1}{\sqrt{3}}]$. If $\tau^1_{AB}$ is entangled, then $\rho_{AB}$ is EPR steerable from Bob to Alice. On the other hand, for any qubit-qudit state $\rho_{AB}$ shared between Alice and Bob, define another new qubit-qudit state $\tau^2_{AB}$ given by
\begin{equation}
\tau^2_{AB} = \mu_2 \ \rho_{AB} + (1-\mu_2) \ \tilde{\rho^2}_{AB},
\end{equation}
where $\tilde{\rho^2}_{AB} = \frac{\mathbb{I}}{2} \otimes \rho_B$ with $\rho_B = \tr_{A}[ \rho_{AB} ] = \tr_{A}[ \tau^2_{AB} ]$ being the reduced state (qudit) at Bob's side; $\mu_2 \in [0, \frac{1}{\sqrt{3}}]$. If $\tau^2_{AB}$ is entangled, then $\rho_{AB}$ is EPR steerable from Alice to Bob.
\label{ob2}
\end{observation} 
Note that Observation \ref{ob2} enables to detect EPR steering from Bob to Alice of a qudit-qubit state shared between Alice and Bob, but it does not enable to detect EPR steering from Alice to Bob of the qudit-qubit state shared between Alice and Bob. Similarly, using Observation \ref{ob2} one can detect EPR steering from Alice to Bob of a qubit-qudit state shared between Alice and Bob, but can not detect EPR steering from Bob to Alice of the qubit-qudit state shared between Alice and Bob. 

One important point to be stressed here is that the proof of Theorem \ref{th1} is based on the fact that the unnormalized conditional state on Bob's side $\varsigma^A_a$ is a $2 \times 2$ matrix mentioned in Eq.(\ref{cons}). Hence, this proof is only applicable when Bob's system is a qubit. That is why the mathematical procedure described in the proof of Theorem \ref{th1} cannot be generalized to detect EPR steering of shared state belonging to the Hilbert space $\mathcal{H}^{d_1} \otimes \mathcal{H}^{d_2}$ (where, $d_1$ and $d_2$ are arbitrary). Further research is needed to investigate whether EPR steering of an arbitrary dimensional state can be detected through entanglement detection using different procedure.

\section{Illustration with examples}\label{sec4}

Any two-qubit pure entangled state is EPR-steerable \cite{in13}. The detection of entanglement of any given two-qubit pure state thus certifies the EPR steerability of the given state. However, this is not true for arbitrary mixed two-qubit states, i. e., there exists two-qubit mixed entangled states, which are unsteerable. Hence, for arbitrary two-qubit mixed states EPR steering cannot be detected by detecting entanglement of that state. The novelty of Theorem \ref{th1} is that it enables to detect EPR steering of any arbitrary two-qubit state (pure as well as mixed) by detecting entanglement of another constructed two-qubit state. In the following we will detect EPR steering of different families of two-qubit mixed states using our proposed Theorem \ref{th1}. These families of states are chosen for various reasons. Some families are chosen as their experimental preparations have been reported. Some other families of states are chosen in order to compare Theorem \ref{th1} with the results obtained by using Semi-definite Program (SDP).  Finally, we will use Observation \ref{ob2} to detect EPR steering of a family of qubit-qutrit mixed states.

$\bullet$ Consider that the following two-qubit Werner state is shared between Alice and Bob.
\begin{equation}
\rho_{AB} = p \ | \psi \rangle \langle \psi | + (1-p) \ \dfrac{\mathbb{I}}{2} \otimes \dfrac{\mathbb{I}}{2},
\label{werner}
\end{equation}
where $| \psi \rangle = \frac{1}{\sqrt{2}} \ (|01\rangle - |10 \rangle)$ is the singlet state, $\{|0\rangle$, $|1\rangle \}$ being the orthonormal basis in $\mathbb{C}^2$, $0 \leq p \leq 1$. Experimental preparation of Werner state via spontaneous parametric down-conversion and controlled decoherence of photons was demonstrated in \cite{expw}. We want to test in which range of $p$ the state $\rho_{AB}$ given by Eq.(\ref{werner}) is steerable from Bob to Alice or that from Alice to Bob using Theorem \ref{th1}. 

Following Theorem \ref{th1}, from the given two-qubit state $\rho_{AB}$ (\ref{werner}) we construct the new two-qubit state given by
\begin{equation}
\tau^1_{AB} = \mu_1 \ \rho_{AB} + (1-\mu_1) \ \tilde{\rho^1}_{AB},
\label{new1}
\end{equation}
where $\tilde{\rho^1}_{AB} = \rho_A \otimes \frac{\mathbb{I}}{2}$ with $\rho_A = \tr_{B}[ \rho_{AB} ] = \frac{\mathbb{I}}{2}$ being the reduced state at Alice's side. We choose $\mu_1 = \frac{1}{\sqrt{3}}$.
The newly constructed state $\tau^1_{AB}$ given by Eq. (\ref{new1}) is entangled for $p > \frac{1}{\sqrt{3}}$ which can be checked using Peres-Horodecki criteria \cite{ph1,ph2}.  Hence, Theorem \ref{th1} concludes that the given two-qubit Werner state $\rho_{AB}$ (\ref{werner}) is steerable from Bob to Alice for $p > \frac{1}{\sqrt{3}}$.

In a similar way, using Theorem \ref{th1} one can show that the given two-qubit Werner state $\rho_{AB}$ (\ref{werner}) is steerable from Alice to Bob for $p > \frac{1}{\sqrt{3}}$.

Note that using quantum violation of $2$-settings linear steering inequality proposed in \cite{in4} (we have used the particular form of this inequality mentioned in \cite{in17}), Werner state (\ref{werner}) is steerable (from Alice to Bob and from Bob to Alice) for $p > \frac{1}{\sqrt{2}}$. On the other hand, Werner state is both-way steerable for $p > \frac{1}{\sqrt{3}}$ using $3$-settings linear steering inequality \cite{in4,in17}. Hence, in this case our proposed technique to check EPR steerability through entanglement detection provides advantage with respect to $2$-settings linear steering inequality. However,  our proposed technique and $3$-settings linear steering inequality detect steerability of Werner state in the same region. Unsteerability of Werner state has been demonstrated by constructing LHS model using SDP in the region $p \leq \frac{1}{2}$ \cite{nin18}.

One important point to be stressed here is that the converse of Theorem \ref{th1} is not always true, i. e., if $\tau^1_{AB}$ is separable, then $\rho_{AB}$ may or may not be EPR steerable from Bob to Alice. On the other hand, if $\tau^2_{AB}$ is separable, then $\rho_{AB}$ may or may not be EPR steerable from Alice to Bob. From the above example it can be checked using Peres-Horodecki criteria \cite{ph1,ph2} that the newly constructed state $\tau^1_{AB}$ (\ref{new1}) (with $\mu_1 = \frac{1}{\sqrt{3}}$) is separable for $p \leq \frac{1}{\sqrt{3}}$. However, the given two-qubit Werner state $\rho_{AB}$ (\ref{werner}) is steerable for $p > \frac{1}{2}$ \cite{steer,nin18}. Hence, in the region $\frac{1}{2} < p \leq \frac{1}{\sqrt{3}}$ the newly constructed two-qubit state $\tau^1_{AB}$ (\ref{new1}) is separable, but the given two-qubit state $\rho_{AB}$ (\ref{werner}) is steerable.

$\bullet$ We will now check EPR steerability of a class of maximally entangled mixed state (MEMS) proposed by Munro \textit{et al.}  \cite{munro}.  MEMS are those states that achieve the greatest possible entanglement for a given mixedness. The MEMS proposed by Munro \textit{et al.} is given by
\begin{align}
\rho_{munro} \ &= \ \begin{pmatrix}
h(C) && 0 && 0 && \dfrac{C}{2}\\[10pt]
0 && 1- 2 h(C) && 0 && 0\\[10pt]
0 && 0 && 0 && 0\\[10pt]
\dfrac{C}{2} && 0 && 0 && h(C)\\
\end{pmatrix},
\label{munro}
\end{align}
where 
\begin{align}
&h(C) \ = \ 
\begin{dcases}
    \dfrac{1}{3},& \text{if } C < \dfrac{2}{3}  \\
    \dfrac{C}{2},& \text{if  } C \geq \dfrac{2}{3}
\end{dcases},
\label{sum}
\end{align}
with $C$ denoting the concurrence of the state $\rho_{munro}$ (\ref{munro}). Experimental technique to prepare this state using correlated photons from parametric down-conversion has been presented \cite{expmunro}.

Munro state $\rho_{munro}$ (\ref{munro}) demonstrates both-way EPR steerability for $C > 0.531$ following Theorem \ref{th1}. On the other hand, Munro state demonstrates both-way steering for $C > 0.707$ and for $C > 0.667$ using quantum violations of $2$-settings linear steering inequality and $3$-settings linear steering inequality, respectively. Therefore, in the region $0.531 < C \leq 0.667$, steerability of Munro state $\rho_{munro}$ (\ref{munro}) can be detected using Theorem \ref{th1}, but not using $2$-settings linear steering inequality and $3$-settings linear steering inequality.

$\bullet$ We now focus on a class of non-maximally entangled mixed states (NMEMS). The states, which are not MEMS, are called NMEMS. In particular, we investigate EPR steering of the Werner derivative states \cite{hiro} which can be obtained by applying a nonlocal unitary operator on the Werner state. Werner derivative state is given by
\begin{equation}
\rho_{wd} = \alpha \ | \psi_{\theta} \rangle \langle \psi_{\theta} | + (1 - \alpha) \ \dfrac{\mathbb{I}}{2} \otimes \dfrac{\mathbb{I}}{2},
\label{wd}
\end{equation}
where $| \psi_{\theta} \rangle$ = $ \cos \theta \ | 00 \rangle + \sin \theta \ | 11 \rangle$ with $0 \leq \theta \leq \frac{\pi}{4}$ and $0 \leq \alpha \leq 1$. The state $\rho_{wd}$ is entangled for $\alpha > [1 + 2 \sin (2\theta)]^{-1}$. In Fig. \ref{fig1} we have shown the region of $\alpha$ and $\theta$ for which both-way EPR steering of Werner derivative state (\ref{wd}) is detected using Theorem \ref{th1}, quantum violation of $2$-settings linear steering inequality and quantum violation of $3$-settings linear steering inequality, respectively. From this Figure it is clear that Theorem \ref{th1} detects steerability of Werner derivative state for a larger region of $\alpha$ and $\theta$ than $2$-settings linear steering inequality and $3$-settings linear steering inequality. The unsteerability of the state $\rho_{wd}$ given by Eq. (\ref{wd}) was analyzed using SDP in Ref. \cite{nin18}. The range of the state parameters $\alpha$ and $\theta$ where Werner derivative state (\ref{wd}) demonstrates EPR steering was also presented in \cite{nin18} using 9 projective measurements and the SDP method of Refs. \cite{nin10,in13}. By comparing the result presented in \cite{nin18} with Fig. \ref{fig1}, it can be checked that SDP technique adapted in \cite{nin18,nin10,in13} demonstrates EPR steering of Werner derivative state (\ref{wd}) for a larger region of state parameters compared to the technique presented in Theorem \ref{th1}. For example, from Fig. \ref{fig1} it is evident that Theorem \ref{th1} cannot detect EPR steering of the state $\rho_{wd}$ with $\alpha = 0.55$ for any values of $\theta$. However, SDP technique adapted in \cite{nin18,nin10,in13} can detect EPR steering of the state $\rho_{wd}$ with $\alpha = 0.55$ for some values of $\theta$.

 \begin{figure}[t!]
 \centering
\includegraphics[width=8.5cm,height=8cm]{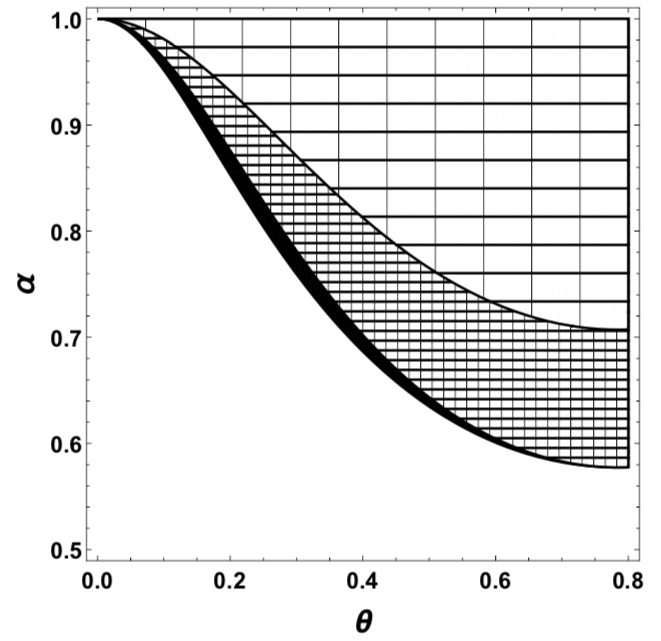}
\caption{Black region is denoted by `Region 1'. The region with small square is denoted by `Region 2'. The region with large square is denoted by `Region 3'. Theorem \ref{th1} detects EPR steerability of the Werner derivative state given by Eq. (\ref{wd}) for the values of state parameters $\alpha$ and $\theta$ indicated by Region 1 or Region 2 or Region 3. Quantum violation of $3$-settings linear steering inequality \cite{in4,in17} detects EPR steerability of the Werner derivative state (\ref{wd}) for the values of $\alpha$ and $\theta$ indicated by Region 2 or Region 3, but not by Region 1. Quantum violation of $2$-settings linear steering inequality \cite{in4,in17} detects EPR steerability of the Werner derivative state (\ref{wd}) for the values of $\alpha$ and $\theta$ indicated by Region 3, but not by Region 1 and Region 2.}\label{fig1}
\end{figure}

$\bullet$ Let us consider another class of NMEMS proposed in \cite{adhikary} given by
\begin{align}
\rho_{p} \ = \ \begin{pmatrix}
\dfrac{p+2}{6} && 0 && 0 && 0 \\[10pt]
0 && \dfrac{1-p}{3} && \dfrac{1-p}{3} && 0 \\[10pt]
0 && \dfrac{1-p}{3} && \dfrac{1-p}{3} && 0 \\[10pt]
0 && 0 && 0 && \dfrac{p}{2} \\
\end{pmatrix},
\label{newnmems}
\end{align}
where $0\leq p \leq 1$. Note that the above state can be written as
\begin{equation}
\rho_{p} = p \ \rho^G_{AB} + (1-p) \ \rho^W_{AB},
\end{equation}
where $\rho^G_{AB} = \tr_{C} [ \ |GHZ \rangle \langle GHZ| \ ]$ and $\rho^W_{AB} = \tr_{C} [|W \rangle \langle W|]$. $|GHZ \rangle$ and $|W \rangle$ are the GHZ (Greenberger-Horne-Zeilinger) state and W state given by,
\begin{equation}
|GHZ \rangle = \frac{1}{\sqrt{2}} \ (|000\rangle + |111 \rangle),
\end{equation}
and 
\begin{equation}
|W \rangle = \frac{1}{\sqrt{3}} \ (|001\rangle + |010 \rangle + |100 \rangle).
\end{equation}
The state (\ref{newnmems}) has experimental relevance since it has been constructed from two qubit GHZ and W states, and GHZ as well as W states are both experimentally realizable in photonic systems using spontaneous parametric down-conversion \cite{expnmems1,expnmems2}. 

Theorem \ref{th1} detects both-way EPR steerability of $\rho_{p}$ given by Eq. (\ref{newnmems}) for $p < 0.073$. On the other hand, EPR steerability of the above state is not detected using $2$-settings linear steering inequality and $3$-settings linear steering inequality.

$\bullet$ Now we will study EPR steering of the following class of maximally steerable mixed state (MSMS) (the states that violate to the most degree a steering inequality for a given mixedness) proposed by Ren \textit{et al.} \cite{ren},
\begin{align}
\rho_{\tau} \ = \ \begin{pmatrix}
\dfrac{1- \tau}{4} && 0 && 0 && \dfrac{1- \tau}{4} \\[10pt]
0 && \dfrac{1+ \tau}{4} && \dfrac{1+ \tau}{4} && 0 \\[10pt]
0 && \dfrac{1+ \tau}{4} && \dfrac{1+ \tau}{4} && 0 \\[10pt]
\dfrac{1- \tau}{4} && 0 && 0 && \dfrac{1 - \tau}{4} \\
\end{pmatrix},
\label{msms}
\end{align}
where $-1 \leq \tau \leq 1$. Note that the above state can be written in the following form
\begin{equation}
\rho_{\tau} = \frac{1-\tau}{2} \ \rho_1 + \frac{1+\tau}{2} \ \rho_2,
\end{equation}
where $\rho_i = |\psi_i \rangle \langle \psi_i|$ with $|\psi_1 \rangle$ = $\frac{1}{\sqrt{2}} \ (|00 \rangle + |11 \rangle)$ and $|\psi_2 \rangle$ = $\frac{1}{\sqrt{2}} \ (|01 \rangle + |10 \rangle)$. Hence, the state $\rho_{\tau}$ can be prepared as a probabilistic mixture of two Bell states.

$2$-settings linear steering inequality as well as $3$-settings linear steering inequality detect both-way steerability of the state $\rho_{\tau}$ for $-1 \leq \tau \leq 1$. Theorem \ref{th1} proposed in this study detects both-way steerability of the state $\rho_{\tau}$ for $-1 \leq \tau <  -0.366$ and for $0.366 < \tau \leq 1$. Hence, in this case $2$-settings linear steering inequality and $3$-settings linear steering inequality are more useful in detecting EPR steerability than our proposed theorem.

$\bullet$ Let us consider the following class of state, which demonstrate one-way steerability \cite{one_way_steer}, is shared between Alice and Bob
\begin{align}
\rho_{\alpha} = \alpha \ | \psi \rangle \langle \psi | + \dfrac{1- \alpha}{5} \ \Big( 2 \ |0\rangle \langle 0| \otimes \dfrac{\mathbb{I}}{2} + 3 \ \dfrac{\mathbb{I}}{2} \otimes |1\rangle \langle 1| \Big),
\label{ows}
\end{align}
where $| \psi \rangle = \frac{1}{\sqrt{2}} \ (|01\rangle - |10 \rangle)$ and $0 \leq \alpha \leq 1$. $2$-settings linear steering inequality detects both-way EPR steerability of the above state $\rho_{\alpha}$ for $\alpha > \frac{1}{\sqrt{2}}$ = $0.707$. $3$-settings linear steering inequality detects both-way EPR steerability of the above state for $\alpha > \frac{1}{\sqrt{3}}$ = $0.577$. Our proposed Theorem \ref{th1} detects EPR steerability from Bob to Alice of the above state $\rho_{\alpha}$ for $\alpha > 0.577$. Interestingly, Theorem \ref{th1} detects EPR steerability from Alice to Bob of the state $\rho_{\alpha}$ (\ref{ows}) for $\alpha > 0.566$. It is to be noted that the state $\rho_{\alpha}$ (\ref{ows}) demonstrates one-way steering in the range $0.4983 \lesssim \alpha \leq \frac{1}{2}$ \cite{one_way_steer}. In particular, it has been shown that the state $\rho_{\alpha}$ (\ref{ows}) does not demonstrate EPR steering from Bob to Alice for $\alpha \leq \frac{1}{2}$. On the other hand, using SDP it has been shown that the above state is steerable from Alice to Bob for $\alpha \gtrsim 0.4983$. Hence, SDP used in Ref. \cite{one_way_steer} is more efficient than our proposed Theorem \ref{th1} in detecting EPR steering of the state $\rho_{\alpha}$ from Alice to Bob.

$\bullet$ Consider that an initial maximally entangled two-qubit state $|\psi_i \rangle$ $=$ $\frac{1}{\sqrt{2}} \ (|00\rangle + |11\rangle)$, shared between Alice and Bob, is subjected to independent local amplitude damping channel given by the evolution $\rho_f$ = $\sum_{i=1}^2 \sum_{j=1}^{2} E_i \otimes E_j \ \rho_i \ E_i^{\dagger} \otimes E_j^{\dagger}$. Here $\rho_i$ = $|\psi_i \rangle \langle \psi_i |$ is the density matrix of the initially shared state, $\rho_f$ is the density matrix of the final state. The two Kraus operators are defined as $E_0$ = $ |0 \rangle \langle 0| + \sqrt{1-p} \ |1 \rangle \langle 1|$ and $E_1$ = $  \sqrt{p} \ |0 \rangle \langle 1|$ with $0 \leq p \leq 1$. Note that experimental technique has already been realized to engineer amplitude damping channel \cite{adc}. $2$-settings linear steering inequality detects both-way EPR steerability of the final state for $0 \leq p < 0.293$. $3$-settings linear steering inequality detects both-way EPR steerability of the final state for $0 \leq p < 0.397$. Our proposed Theorem \ref{th1} detects both-way EPR steerability of the above final state  for $0 \leq p < 0.411$. Hence, in this case $2$-settings linear steering inequality and $3$-settings linear steering inequality are less useful in detecting EPR steerability than our proposed theorem. Importantly, it has been shown \cite{nin17} using SDP that the above final state has LHS model for $p \gtrsim 0.4$. Hence, the approximate lower bound of $p$ where the final state has LHS model obtained using SDP in Ref. \cite{nin17} almost matches with the upper bound of $p$ where our proposed theorem detects EPR steering of the aforementioned final state.

$\bullet$ Let us focus on checking EPR-steerability of some higher dimensional quantum system. In particular, let us consider the following qubit-qutrit state shared between Alice and Bob,
\begin{equation}
\rho_\mu = (1-\mu) \rho_w + \mu \frac{\mathbb{I}}{2} \otimes |v \rangle \langle v |,
\label{newhigh}
\end{equation}
where $0 \leq \mu \leq 1$; $\rho_w$ = $p \ | \psi \rangle \langle \psi | + (1-p) \ \dfrac{\mathbb{I}}{2} \otimes \dfrac{\mathbb{I}}{2}$ is the Werner state mentioned in Eq.(\ref{werner}) with $0 \leq p \leq 1$; $| \psi \rangle = \frac{1}{\sqrt{2}} \ (|01\rangle - |10 \rangle)$ is the singlet state. $|v\rangle$ is a vacuum state orthogonal to Bob's qubit subspace. The above state (\ref{newhigh}) can be prepared when one subsystem of Werner state (\ref{werner}) is distributed through a lossy channel to Bob. Note that a lossy channel is one that replaces a qubit with the vacuum state $|v\rangle$ with probability $\mu$ and can be represented by the map $\rho$ $\rightarrow$ $(1-\mu) \rho + \mu |v \rangle \langle v |$. We will now check EPR steerability of the qubit-qutrit state (\ref{newhigh}) following Observation \ref{ob2}. Hence, we construct the following new qubit-qutrit state from the state $\rho_{\mu}$ mentioned in Eq. (\ref{newhigh}),
\begin{equation}
\tau = \frac{1}{\sqrt{3}} \rho_\mu + \Big(1-\frac{1}{\sqrt{3}}\Big) \tilde{\rho_\mu},
\label{newh2}
\end{equation}
where $\tilde{\rho_\mu}$ = $\frac{\mathbb{I}}{2} \otimes \rho_B$ with $\rho_B$ = $\tr_A[\rho_\mu]$ being the reduced state (qutrit) of $\rho_{\mu}$ at Bob's state. The qubit-qutrit state (\ref{newh2}) is entangled for $\frac{1}{\sqrt{3}} < p \leq 1$ for any values of $\mu$. This can be checked through Peres-Horodecki criteria \cite{ph1,ph2}. Hence, following Observation \ref{ob2} we can conclude that the qubit-qutrit state mentioned in Eq. (\ref{newhigh}) is EPR steerable from Alice to Bob for $\frac{1}{\sqrt{3}} < p \leq 1$ for any values of $\mu$.

It has been shown \cite{newpap1,newpap2} that the qubit-qutrit state (\ref{newhigh}) is EPR steerable from Alice to Bob for $\frac{1}{2} < p \leq 1$ for any values of $\mu$. Hence, the technique adopted in Refs. \cite{newpap1,newpap2} is more useful than Observation \ref{ob2} in detecting EPR steerability of the qubit-qutrit state (\ref{newhigh}) from Alice to Bob. One important point to be stressed here is that the state (\ref{newhigh}) is one-way EPR steerable since for $\frac{1}{2} < p \leq 1$ the state (\ref{newhigh}) is not EPR steerable from Bob to Alice for certain values of $\mu$ considering projective measurements or positive operator valued measurements (POVM) \cite{newpap1,newpap2}. On the other hand, our proposed Observation \ref{ob2} cannot detect EPR steering of the qubit-qutrit state (\ref{newhigh}) from Bob to Alice. Hence, we cannot use Observation \ref{ob2} to demonstrate one-way EPR steering of the qubit-qutrit state (\ref{newhigh}).

As mentioned earlier, the experimental preparations of some of the aforementioned families of two-qubit states have been reported. Hence, after preparing these states following the proposed techniques one can use local depolarizing channel \cite{dc} and structural physical approximation to partial transposition \cite{spa3} as mentioned in Section \ref{sec3} in order to experimentally detect EPR steering of these families of states using our proposed Observation \ref{obbb2}.

\section{Concluding discussions}\label{sec5}

Since EPR steering has foundational significance as well as information theoretic applications, detecting EPR steering is one of the most profound problem in recent times. A number of criteria to detect EPR steering has been proposed till date \cite{in1,in4,in5,in6,in10,in8,nin18,nin10,in13,in17,nin17}. In the present study we have provided a novel criteria to detect EPR steering of an arbitrary two-qubit state. This criteria enables one to detect EPR steering of the given two-qubit state by detecting entanglement of another constructed two-qubit state. Hence, theoretically one can detect EPR steering of the given state using Peres-Horodecki criteria without using any steering inequality. Besides having foundational importance, our proposed technique to detect EPR steering through entanglement detection reduces the ``complexity cost"  in experimentally determining EPR steering as the ``complexity cost" for the least complex demonstration of entanglement is less than the ``complexity cost" for the least complex demonstration of EPR steering \cite{cc}. Moreover, this study may pave a new way in avoiding locality loophole in detecting EPR steering experimentally.

It is to be noted that any quantum state, which is EPR steerable, is entangled as well, since quantum states demonstrating EPR steering  form a strict subset of the entangled states \cite{inequi}. Hence, entanglement can be detected by detecting EPR steering. Although, the converse is not true always as there exist entangled unsteerable states \cite{inequi}. For this reason, in general, detecting entanglement in a quantum state does not guarantee that the state under consideration is EPR steerable. However, the novelty of the present study is that it provides an indirect way to detect EPR steering though entanglement detection. Previously, it was shown that Bell nonlocality can be indirectly detected by detecting EPR steering \cite{pati}. Hence, the present study together with the results obtained in \cite{pati} connect three inequivalent forms of quantum inseparabilities.

The present study is restricted to only qubit-qubit, qubit-qudit and qudit-qubit systems. Whether EPR steering can be detected through entanglement detection in arbitrary higher dimensional system (qudit-qudit) as well as in multipartite scenario is worth to be studied in future. In the present study we have discussed some possible outlines to experimentally detect EPR steering through our proposed criteria using local depolarizing channel together with entanglement witness (see, for example, \cite{ew1,ew6} and the references therein) or ``structural physical approximation" to the partial transpose operation \cite{spa1,spa3}. Hence, the experimental realization of our proposed criteria following the above ideas is another area for future research. 

\section*{Acknowledgements}
The authors acknowledge the anonymous referees for valuable comments. D. D. and S. S. acknowledge fruitful discussions with Arkaprabha Ghosal, Som Kanjilal and Arup Roy. D. D. acknowledges the financial support from University Grants Commission (UGC), Government of India. S. S. acknowledges the financial support from INSPIRE programme, Department of Science and Technology, Government of India.\\

\textit{Note added-} After communicating this manuscript to the journal, a related work by Changbo Chen \textit{et al.} \cite{new} appeared in arXiv.

\end{document}